\documentclass{llncs}

\usepackage{hyperref}
\usepackage{amssymb}
\usepackage{graphicx}
\usepackage{float}
\usepackage{comment}
\usepackage{subfig}
\usepackage{footnote}
\usepackage{complexity}
\usepackage[toc,page]{appendix}
\usepackage{placeins}
\usepackage{environ,etoolbox}

\DeclareGraphicsExtensions{.pdf}

\newif\ifappendix
\NewEnviron{maybeappendix}[1]
    {\ifcsname putmaybeappendix#1\endcsname
    \else
        \expandafter\gdef\csname putmaybeappendix#1\endcsname{}
    \fi
    \ifappendix
        \expandafter\xdef\csname putmaybeappendix#1\endcsname{\expandafter\expandonce\csname putmaybeappendix#1\endcsname\unexpanded{\par}\expandonce\BODY}
    \else
        \BODY%
    \fi}
\newcommand{\putmaybeappendix}[1]{\csname putmaybeappendix#1\endcsname}

\appendixfalse

\begin{document}
\title{Improved Lower Bounds for \\ Graph Embedding Problems}

\author{Hans L. Bodlaender\inst{1,2}\thanks{The research of this author was partially supported by the NETWORKS project,
funded by the Netherlands Organization for Scientific Research NWO.} \and Tom C. van der Zanden\inst{1}}


\institute{Department of Computer Science, Utrecht University, Utrecht, The Netherlands \email{\{H.L.Bodlaender,T.C.vanderZanden\}@uu.nl} \and Department of Mathematics and Computer Science, Eindhoven University of Technology, Eindhoven, The Netherlands}

\newlength{\problemoffset}
\setlength{\problemoffset}{0in}

\newcommand{\decision}[3]{
\begin{list}{}{
\setlength{\leftmargin}{\problemoffset}
\setlength{\rightmargin}{\problemoffset}
\setlength{\parsep}{0pt}
\setlength{\itemsep}{2pt}
\setlength{\topsep}{\itemsep}
\setlength{\partopsep}{\itemsep}
}
\item
{\textsc{#1}}
\item
{\textbf{Instance:} #2}
\item
{\textbf{Question:} #3}
\end{list}
}

\maketitle

\begin{abstract}
In this paper, we give new, tight subexponential lower bounds for a number of graph embedding problems. We introduce two related combinatorial problems, which we call {\sc String Crafting} and {\sc Orthogonal Vector Crafting}, and show that these cannot be solved in time $2^{o(|s|/\log{|s|})}$, unless the Exponential Time Hypothesis fails.

These results are used to obtain simplified hardness results for several graph embedding problems, on more restricted graph classes than previously known: assuming the Exponential Time Hypothesis, there do not exist algorithms that run in $2^{o(n/\log n)}$ time for 
{\sc Subgraph Isomorphism} on graphs of pathwidth 1,
{\sc Induced Subgraph Isomorphism} on graphs of pathwidth 1,
{\sc Graph Minor} on graphs of pathwidth 1,
{\sc Induced Graph Minor} on graphs of pathwidth 1,
{\sc Intervalizing 5-Colored Graphs} on trees,
and finding a tree or path decomposition with width at most $c$ with a minimum number of bags, for any fixed $c\geq 16$.

$2^{\Theta(n/\log n)}$ appears to be the ``correct'' running time for many packing and embedding problems on restricted graph classes, and we think {\sc String Crafting} and {\sc Orthogonal Vector Crafting} form a useful framework for establishing lower bounds of this form.
\end{abstract}

\section{Introduction}

Many $\NP$-complete graph problems admit faster algorithms when restricted to planar graphs. In almost all cases, these algorithms have running times that are exponential in a square root function (e.g. $2^{O(\sqrt{n})}$, $n^{O(\sqrt{k})}$ or $2^{O(\sqrt{k})} n^{O(1)}$) and most of these results are tight, assuming the Exponential Time Hypothesis. This seemingly universal behaviour has been dubbed the ``Square Root Phenomenon'' \cite{rootphenomenon}. The open question \cite{marx-open} of whether the Square Root Phenomenon holds for {\sc Subgraph Isomorphism} in planar graphs, has recently been answered in the negative: assuming the Exponential Time Hypothesis, there is no $2^{o(n/\log{n})}$-time algorithm for {\sc Subgraph Isomorphism}, even when restricted to (planar) graphs of pathwidth $2$ \cite{icalppaper}. The same lower bound holds for {\sc Induced Subgraph} and {\sc (Induced) Minor} and is in fact tight: the problems admit $2^{O(n/\log{n})}$-time algorithms on $H$-minor free graphs \cite{icalppaper}.

The lower bounds in \cite{icalppaper} follow by reductions from a problem called {\sc String 3-Groups}. We introduce a new problem, {\sc String Crafting}, and establish a $2^{\Omega(|s|/\log{|s|})}$-time lower bound under the ETH for this problem by giving a direct reduction from {\sc 3-Satisfiability}. Using this result, we show that the $2^{\Omega(|s|/\log{|s|})}$-time lower bounds for (Induced) Subgraph and (Induced) Minor hold even on graphs of pathwidth $1$.

Alongside {\sc String Crafting}, we introduce the related {\sc Orthogonal Vector Crafting} problem. Using this problem, we show $2^{\Omega(|n|/\log{|n|})}$-time lower bounds for deciding whether a $5$-coloured tree is the subgraph of an interval graph (for which the same colouring is proper) and for deciding whether a graph admits a tree (or path) decomposition of width $16$ with at most a given number of bags.

For any fixed $k$, {\sc Intervalizing $k$-Coloured Graphs} can be solved in time $2^{O(n/\log n)}$ \cite{intervalizing-exact}. Bodlaender and Nederlof \cite{mspd-springer} conjecture a lower bound (under the Exponential Time Hypothesis) of $2^{\Omega(n/\log n)}$ time for $k\geq 6$; we settle this conjecture and show that it in fact holds for $k\geq 5$, even when restricted to trees. To complement this result for a fixed number of colours, we also show that there is no algorithm solving {\sc Intervalizing Coloured Graphs} (with an arbitrary number of colours) in time $2^{o(n)}$, even when restricted to trees.

The minimum size path and tree decomposition problems can also be solved in $2^{O(n/\log n)}$ time on graphs of bounded treewidth. This is known to be tight under the Exponential Time Hypothesis for $k\geq 39$ \cite{mspd-springer}. We improve this to $k\geq 16$; our proof is also simpler than that in \cite{mspd-springer}.

Our results show that {\sc String Crafting} and {\sc Orthogonal Vector Crafting} are a useful framework for establishing lower bounds of the form $2^{\Omega(n/\log n)}$ under the Exponential Time Hypothesis. It appears that for many packing and embedding problems on restricted graph classes, this bound is tight.

\section{Preliminaries}

\paragraph*{Strings.} We work with the alphabet $\{0,1\}$; i.e., strings are elements of $\{0,1\}^\ast$.
The length of a string $s$ is denoted by $|s|$. The $i^\textrm{th}$ character of a string $s$ is denoted by $s(i)$.
Given a string $s\in \{0,1\}^\ast$, $\overline{s}$ denotes binary complement of $s$, that is, each occurrence of a $0$ is replaced by a $1$ and vice versa; i.e., $|s|=|\overline{s}|$, and for $1\leq i \leq |s|$, $\overline{s}(i) = 1 - s(i)$.
E.g., if $s= 100$, then $\overline{s} = 011$.
With $s^R$, we denote the string $s$ in reverse order; e.g., if $s=100$, then $s^R = 001$.
The concatenation of strings $s$ and $t$ is denoted by $s \cdot t$.
A string $s$ is a {\em palindrome}, when $s=s^R$. By $0^n$ (resp. $1^n$) we denote the string that consists of $n$ $0$'s (resp. $1$'s).

\paragraph*{Graphs.}
Given a graph $G$, we let $V(G)$ denote its vertex set and $E(G)$ its edge set. Let $Nb(v)$ denote the open neighbourhood of $v$, that is, the vertices adjacent to $v$, excluding $v$ itself.

\paragraph*{Treewidth and Pathwidth.} A \emph{tree decomposition} of a graph $G=(V,E)$ is a tree $T$ with vertices $t_1,\ldots,t_s$ with for each vertex $t_i$ a \emph{bag} $X_i\subseteq V$ such that for all $v\in V$, the set $\{t_i\in\{t_1,\ldots,t_s\} \mid v\in X_i\}$ is non-empty and induces a connected subtree of $T$ and for all $(u,v)\in E$ there exists a bag $X_i$ such that $\{u,v\}\in X_i$. The \emph{width} of a tree decomposition is $\max_i \{|X_i|-1\}$ and the \emph{treewidth} of a graph $G$ is the minimum width of a tree decomposition of $G$. A \emph{path decomposition} is a tree decomposition where $T$ is a path, and the pathwidth of a graph $G$ is the minimum width of a path decomposition of $G$.

A graph is a \emph{caterpillar tree} if it is connected and has pathwidth $1$.

\paragraph*{Subgraphs and Isomorphism.}
$H$ is a \emph{subgraph} of $G$ if $V(H)\subseteq V(G)$ and $E(H)\subseteq E(G)$; we say the subgraph is \emph{induced} if moreover $E(H) = E(G)\cap V(H)\times V(H)$. We say a graph $H$ is \emph{isomorphic} to a graph $G$ if there is a bijection $f:V(H)\to V(G)$ so that $(u,v)\in E(H) \iff (f(u),f(v))\in E(G)$. 

\paragraph*{Contractions, minors.}
We say a graph $G'$ is obtained from $G$ by \emph{contracting} edge $(u,v)$, if $G'$ is obtained from $G$ by replacing vertices $u,v$ with a new vertex $w$ which is made adjacent to all vertices in $Nb(u) \cup Nb(v)$. A graph $G'$ is a minor of $G$ if a graph isomorphic to $G'$ can be obtained from $G$ by contractions and deleting vertices and/or edges. $G'$ is an \emph{induced minor} if we can obtain it by contractions and deleting vertices (but not edges).

We say $G'$ is an \emph{$r$-shallow minor} if $G'$ can be obtained as a minor of $G$ and any subgraph of $G$ that is contracted to form some vertex of $G'$ has radius at most $r$ (that is, there is a central vertex within distance at most $r$ from any other vertex in the subgraph). 
Finally, $G'$ is a \emph{topological minor} if we can subdivide the edges of $G'$ to obtain a graph $G''$ that is isomorphic to a subgraph of $G$ (that is, we may repeatedly take an edge $(u,v)$ and replace it by a new vertex $w$ and edges $(u,w)$ and $(w,v)$).

For each of (induced) subgraph, induced (minor), topological minor and shallow minor, we define the corresponding decision problem, that is, to decide whether a pattern graph $P$ is isomorphic to an (induced) subgraph/(induced) minor/topological minor/shallow minor of a host graph $G$.

\section{String Crafting and Orthogonal Vector Crafting}

We now formally introduce the {\sc String Crafting} problem:

\begin{verse}
{\sc String Crafting}\\
{\bf Given:} String $s$, and $n$ strings $t_1, \ldots, t_n$, with $|s| = \sum_{i=1}^n |t_i|$. \\
{\bf Question:} Is there a permutation $\Pi: \{1, \ldots, n\} \rightarrow \{1, \ldots, n \}$, such that
the string $t^\Pi = t_{\Pi(1)} \cdot t_{\Pi(2)} \cdots t_{\Pi}$ fulfils that for each $i$, $1\leq i\leq |s|$,
$s(i) \geq t^{\Pi}(i)$.
\end{verse}

I.e., we permute the collection of strings $\{t_1,t_2,\ldots t_n\}$, then concatenate these, and should obtain a resulting string $t^\Pi$ (that necessarily has the same length as $s$) such that on every position where $t^\Pi$ has a 1, $s$ also has a 1.

Given $\Pi$, $1\leq i\leq |s|$, we let $idx_\Pi(i) = \max \{1\leq j\leq n : \Sigma_{k=1}^j |t_k| \geq i\}$ and let $pos_\Pi(i) = i - \Sigma_{k=1}^{idx_\Pi(i)-1} |t_k|$.

We also introduce the following variation of {\sc String Crafting}, where, instead of requiring that whenever $t^\Pi$ has a $1$, $s$ has a $1$ as well, we require that whenever $t^\Pi$ has a $1$, $s$ has a $0$ (i.e. the strings $t^\Pi$ and $s$, viewed as vectors, are orthogonal).

\begin{verse}
{\sc Orthogonal Vector Crafting}\\
{\bf Given:} String $s$, and $n$ strings $t_1, \ldots, t_n$, with $|s| = \sum_{i=1}^n |t_i|$. \\
{\bf Question:} Is there a permutation $\Pi: \{1, \ldots, n\} \rightarrow \{1, \ldots, n \}$, such that
the string $t^\Pi = t_{\Pi(1)} \cdot t_{\Pi(2)} \cdots t_{\Pi}$ fulfils that for each $i$, $1\leq i\leq |s|$,
$s(i) \cdot t^{\Pi}(i) = 0$, i.e., when viewed as vectors, $s$ is orthogonal to $t^\Pi$.
\end{verse}

\begin{theorem}
Suppose the Exponential Time Hypothesis holds. Then there is no algorithm that solves the {\sc String Crafting} problem in
$2^{o(|s|/\log{|s|})}$ time, even when all strings $t_i$ are palindromes and start and end with a $1$.
\end{theorem}

\begin{proof}
Suppose we have an instance of {\sc 3-Satisfiability} with $n$ variables and $m$ clauses. 
We number the variables $x_1$ to $x_n$ and for convenience, we number the clauses $C_{n+1}$ to $C_{n+m+1}$.

We assume by the sparsification lemma that $m=O(n)$ \cite{sparsification}.

Let $q = \lceil \log (n+m) \rceil$, and let $r =  4 q +2 $.
We first assign an $r$-bit number to each variable and clause; more precisely, we give a mapping $id : \{1, \ldots, n+m\} \rightarrow
\{0,1\}^{r}$. Let $nb(i)$ be the $q$-bit binary representation of $i$, such that $0 \leq nb(i) \leq 2^r-1$.
We set, for $1\leq i \leq n+m$:
\[ id(i) = 1 \cdot nb(i) \cdot \overline{nb(i)} \cdot \overline{nb(i)}^R \cdot{nb(i)}^R \cdot 1 \]
Note that each $id(i)$ is an $r$-bit string that is a palindrome, ending and starting with a 1.

We first build $s$. $s$ is the concatenation of $2n$ strings, each representing one of the literals.

Suppose the literal $x_i$ appears $c_i$ times in a clause, and the literal $\neg x_i$ appears $d_i$ times in a clause.
Set $f_i = c_i + d_i$. 
Assign the following strings to the pair of literals $x_i$ and $\neg x_i$:
\begin{itemize}

\item $a^{x_i}$ is the concatenation of the id's of all clauses in which $x_i$ appears, followed by $d_i$ copies of the string $1 \cdot 0^{r-2} \cdot 1$.
\item $a^{\neg x_i}$ is the concatenation of the id's of all clauses in which $\neg x_i$ appears, followed by $c_i$ copies of the string $1 \cdot 0^{r-2} \cdot 1$.
\item $b^i = id(i) \cdot a^{x_i} \cdot id(i) \cdot a^{\neg x_i} \cdot id(i)$.
\end{itemize}

Now, we set $s= b^1 \cdot b^2 \cdots b^{n-1} \cdot b^n$.

We now build the collection of strings $t_i$. We have three different types of strings:
%
%
%
\begin{itemize}
\item {\em Variable selection:} For each variable $x_i$ we have one string of length $(f_i + 2) r$ of the form
$ id(i) \cdot 0^{r \cdot f_i} \cdot id(i) $.
\item {\em Clause verification:}  For each clause $C_i$, we have a string of the form $id(i)$.
\item {\em Filler strings:} A filler string is of the form $1 \cdot 0^{r-2} \cdot 1$. We have $n+2m$ filler strings.
\end{itemize}

Thus, the collection of strings $t_i$ consists of $n$ variable selection strings, $m$ clause verification strings, and $n+2m$ filler strings. Notice that each of these strings is a palindrome and ends and starts with a 1.

The idea behind the reduction is that $s$ consists of a list of variable identifiers followed by which clauses a true/false assignment to that variable would satisfy. The variable selection gadget can be placed in $s$ in two ways: either covering all the clauses satisfied by assigning true to the variable, or covering all the clauses satisfied by assigning false. The clause verification strings then fit into $s$ only if we have not covered all of the places where the clause can fit with variable selection strings (corresponding to that we have made some assignment that satisfies the clause).

Furthermore, note that since $\Sigma_{i=1}^n f_i = 3m$, the length of $s$ is $(3n+6m)r$, the combined length of the variable selection strings is $(2n+3m)r$, the combined length of the clause verification strings is $mr$, and the filler strings have combined length $(n+2m)r$.

In the following, we say a string $t_i$ is mapped to a substring $s'$ of $s$ if $s'$ is the substring of $s$ corresponding to the position (and length) of $t_i$ in $t^\Pi$.

\begin{lemma}
The instance of {\sc 3-Satisfiability} is satisfiable, if and only if the constructed instance of {\sc String Crafting} has a solution.
\end{lemma}

\begin{proof}
First, we show the reverse implication. Suppose we have a satisfying assignment to the {\sc 3-Satisfiability} instance. Consider the substring of $s$ formed by $b^i$, which is of the form $id(i) \cdot a^{x_i} \cdot id(i) \cdot a^{\neg x_i} \cdot id(i)$. If in the satisfying assignment $x_i$ is true, we choose the permutation $\Pi$ so that variable selection string $id(i) \cdot 0^{r \cdot f_i} \cdot id(i)$ corresponding to $x_i$ is mapped to the substring $id(i) \cdot a^{\neg x_i} \cdot id(i)$; if $x_i$ is false, we map the variable selection string onto the substring $id(i) \cdot a^{x_i} \cdot id(i)$. A filler string is mapped to the other instance of $id(i)$ in the substring.

Now, we show how the clause verification strings can be mapped. Suppose clause $C_j$ is satisfied by the literal $x_i$ (resp. $\neg x_i$). Since $x_i$ is true (resp. false), the substring $a^{x_i}$ (resp. $a^{\neg x_i}$) of $s$ is not yet used by a variable selection gadget and contains $id(j)$ as a substring, to which we can map the clause verification string corresponding to $C_j$.

Note that in $s$ now remain a number of strings of the form $1 \cdot 0^{r-2} \cdot 1$ and a number of strings corresponding to id's of clauses, together $2m$ such strings, which is exactly the number of filler strings we have left. These can thus be mapped to these strings, and we obtain a solution to the {\sc String Crafting} instance. It is easy to see that with this construction, $s$ has a $1$ whenever the string constructed from the permutation does.

Next, for the forward implication, consider a solution $\Pi$ to the {\sc String Crafting} instance. We require the following lemma:

\begin{lemma}
Suppose that $t_i=id(j)$. Then the substring $w$ of $s$ corresponding to the position of $t_i$ in $t^\Pi$ is $id(j)$.
\end{lemma}

\begin{proof}
Because the length of each string is a multiple of $r$, $w$ is either $id(k)$ for some $k$, or the string $1\cdot 0^{r-2} \cdot 1$. Clearly, $w$ can not be $1\cdot 0^{r-2} \cdot 1$ because the construction of $id(i)$ ensures that it has more than $2$ non-zero characters, so at some position $w$ would have a $1$ where $w'$ does not. Recall that $id(i) = 1 \cdot nb(i) \cdot \overline{nb(i)} \cdot \overline{nb(i)}^R \cdot{nb(i)}^R \cdot 1$. If $j\not = k$, then either at some position $nb(k)$ has a $0$ where $nb(j)$ has a $1$ (contradicting that $\Pi$ is a solution) or at some position $\overline{nb(k)}$ has a $0$ where $\overline{nb(j)}$ has a $1$ (again contradicting that $\Pi$ is a solution). Therefore $j=k$.\qed
\end{proof}

Clearly, for any $i$, there are only two possible places in $t^\Pi$ where the variable selection string $id(i) \cdot 0^{r \cdot f_i} \cdot id(i)$ can be mapped to: either in the place of $id(i) \cdot a^{x_i} \cdot id(i)$ in $s$ or in the place of $id(i) \cdot a^{\neg x_i} \cdot id(i)$, since these are the only (integer multiple of $r$) positions where $id(i)$ occurs in $s$. If the former place is used we set $x_i$ to false, otherwise we set $x_i$ to true.

Now, consider a clause $C_j$, and the place where the corresponding clause verification gadget $id(j) $is mapped to. Suppose it is mapped to some substring of $id(i) \cdot a^{x_i} \cdot id(i) \cdot a^{\neg x_i} \cdot id(i)$. If $id(j)$ is mapped to a substring of $a^{x_i}$ then (by construction of $a^{x_i}$) $x_i$ appears as a positive literal in $C_j$ and our chosen assignment satisfies $C_j$ (since we have set $x_i$ to true). Otherwise, if $id(j)$ is mapped to a substring of $a^{\neg x_i}$ $x_i$ appears negated in $C_j$ and our chosen assignment satisfies $C_j$ (since we have set $x_i$ to false).

We thus obtain a satisfying assignment for the {\sc 3-Satisfiability} instance.\qed
\end{proof}

Since in the constructed instance, $|s| = (3n+6m)r$ and $r=O(\log n), m = O(n)$, we have that $|s|=O(n \log{n})$. A $2^{o(|s|/\log{|s|})}$-time algorithm for {\sc String Crafting} would give a $2^{o(n\log{n} / \log{(n\log n)})} = 2^{o(n)}$-time algorithm for deciding {\sc 3-Satisfiability}, violating the ETH.\qed
\end{proof}

Note that we can also restrict all strings $t_i$ to start and end with a $0$ by a slight modification of the proof.

\begin{theorem}
Assuming the Exponential Time Hypothesis, {\sc Orthogonal Vector Crafting} can not be solved in $2^{o(|s|/\log{|s|})}$ time, even when all strings $t_i$ are palindromes and start and end with a $1$.
\end{theorem}

\begin{proof}
This follows from the result for {\sc String Crafting}, by taking the complement of the string $s$.\qed
\end{proof}

Again, we can also restrict all strings $t_i$ to start and end with a $0$. 

As illustrated by the following theorem, these lower bounds are tight. The algorithm is a simpler example of the techniques used in \cite{icalppaper,intervalizing-exact,mspd-springer}. \ifappendix\else There, the authors use isomorphism tests on graphs, here, we use equality of strings.\fi

\ifappendix \begin{theorem}[Appendix \ref{app:sc-algo}.]\else \begin{theorem}\fi\label{thm:sc-algo}
There exists algorithms, solving {\sc String Crafting} and {\sc Orthogonal Vector Crafting} in $2^{O(|s|/\log{|s|})}$.
\end{theorem}

\begin{maybeappendix}{app:sc-algo}\begin{proof}\ifappendix \emph{(Theorem \ref{thm:sc-algo}.)} \fi
The brute-force algorithm of trying all $n!$ permutations of the strings $t_1,\ldots,t_n$ would take $O(|s|!s)$ time in the worst case. This can be improved to $O(2^{|s|}s^2)$ by simple Held-Karp \cite{held-karp} dynamic programming: for each (multi-)subset $K\subseteq \{t_1,\ldots,t_n\}$ and $l=\Sigma_{t\in K} |t|$ we memoize whether the substring $s(1)\cdots s(l)$ of $s$ together with $K$ forms a positive instance of {\sc String Crafting} (resp. {\sc Orthogonal Vector Crafting}).

The number of such (multi-)subsets $K$ is $2^{|s|}$ in the worst case. However, in this case, each string $t\in K$ is of length $1$ and we can instead store the multiplicity of each string, making for only $O(|s|^2)$ cases (since each string is either $0$ or $1$).

More generally, call a string $t_i$ \emph{long} if $|t_i|\geq \log_2(|s|) / 2$ and \emph{short} otherwise. There are at most $2|s|/\log{|s|}$ long strings, and as such we can store explicitly what subset of the long strings is in $K$ (giving $2^{O(|s|/\log{|s|})}$ cases). Since there are at most $2^{\log{|s|}/2} = \sqrt{|s|}$ distinct short strings, storing the multiplicity of each one contributes at most $|s|^{\sqrt{|s|}}=2^{\sqrt{|s|}\log{|s|}}$ cases. \qed
\end{proof}
\end{maybeappendix}

\section{Lower Bounds for Graph Embedding Problems}

\begin{theorem}\label{thm:sgihard}
Suppose the Exponential Time Hypothesis holds. Then there is no algorithm solving {\sc Subgraph Isomorphism} in $2^{o(n/\log n)}$ time, even if $G$ is a caterpillar tree of maximum degree $3$ or $G$ is connected, planar, has pathwidth $2$ and has only one vertex of degree greater than $3$ and $P$ is a tree.
\end{theorem}

\begin{proof}
By reduction from {\sc String Crafting}. We first give the proof for the case that G is a caterpillar tree of maximum degree 3, We construct $G$ from $s$ as follows: we take a path of vertices $v_1,\ldots,v_{|s|}$ (\emph{path vertices}). If $s(i)=1$, we add a \emph{hair vertex} $h_i$ and edge $(v_i,h_i)$ to $G$ (obtaining a caterpillar tree).
We construct $P$ from the strings $t_i$ by, for each string $t_i$ repeating this construction, and taking the disjoint union of the caterpillars created in this way (resulting in a graph that is a forest of caterpillar trees, i.e., a graph of pathwidth $1$). An example of this construction is depicted in Figure \ref{fig:subgraph-red}.
{\ifappendix The constructed instance of $G$ contains $P$ as a subgraph, if and only if the instance of {\sc String Crafting} has a solution: the order in which the caterpillars are embedded in $G$ gives the permutation of the strings: when a caterpillar in $P$ has a hair, $G$ must have a hair at the specific position, which
implies that a position with a 1 in the constructed string $t$ must be a position where $s$ also has a 1. See Appendix~\ref{app:subgraphs-lab} for details. \fi}
\begin{figure}[h]
    \centering
    \includegraphics[width=0.4\textwidth]{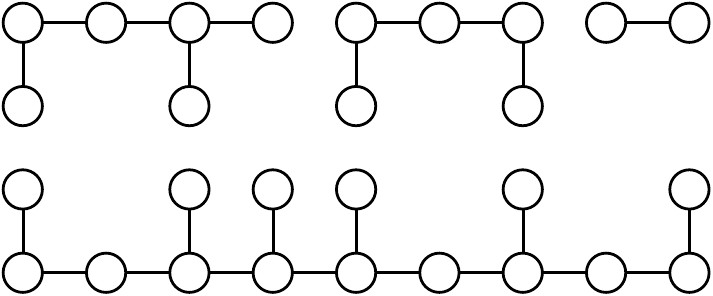}
    \caption{Simplified example of the graphs created in the hardness reduction for Theorem \ref{thm:sgihard}. The bottom caterpillar represents the host graph (corresponding to string $s$), the top caterpillars represent the strings $t_i$ and form the guest graph.
Here $s=101110101$ and $t_1=1010, t_2 = 101$ and $t_3=00$.
}

    \label{fig:subgraph-red}
\end{figure}

\begin{maybeappendix}{app:subgraphs}
\begin{lemma}
The constructed instance of $G$ {\ifappendix in the proof of Theorem~\ref{thm:sgihard} \fi} contains $P$ as a subgraph only if the instance of {\sc String Crafting} has a solution.
\end{lemma}
\begin{proof}
Suppose $P$ contains $G$ as a subgraph. Since $\Sigma_i |t_i| = |s|$ and each string $t_i$ starts and ends with a $1$, the path vertices of $P$ and $G$ must be in one-to-one correspondence (we can not map a hair vertex of $P$ to a path vertex of $G$ since otherwise we would not be able to fit all connected components of $P$ into $G$). The order in which the connected components of $P$ appear as we traverse the path of $G$ gives a permutation $\Pi$ of the strings $t_i$. We claim that this permutation is a solution to the {\sc String Crafting} instance, since $G$ must have a hair vertex whenever $P$ has one (or else we would not have found $P$ as a subgraph) we have that $s$ has a $1$ whenever $t^\Pi$ has a $1$. Note that it does not matter that we can flip a component of $P$ (embed the vertices of the path in the reverse order) since the strings $t_i$ are palindromes.\qed
\end{proof}

\begin{lemma}
The constructed instance of $G$ {\ifappendix in the proof of Theorem~\ref{thm:sgihard} \fi} contains $P$ as a subgraph if the instance of {\sc String Crafting} has a solution.
\end{lemma}
\begin{proof}
Let $\Pi$ be a solution for the {\sc String Crafting} instance. We can map the path vertices of the connected components of $P$ to the path vertices of $G$ in the order the corresponding strings appear in the permutation $\Pi$ (e.g. the first path vertex of the connected component corresponding to $\Pi(1)$ gets mapped to the first path vertex $v_1$, the first path vertex of the component corresponding to $\Pi(2)$ gets mapped to the $|t_{\Pi(1)}|+1^\textrm{th}$ path vertex,...). Whenever a path vertex in $P$ is connected to a hair vertex, $t^\Pi$ has a $1$ in the corresponding position, and therefore $s$ has a $1$ in the corresponding position as well and thus $G$ also has a hair vertex in the corresponding position. We can thus appropriately map the hair vertices of $P$ to the hair vertices of $G$, and see that $P$ is indeed a subgraph of $G$. \qed
\end{proof}
\end{maybeappendix}

Since the constructed instance has $O(|s|)$ vertices, this establishes the first part of the lemma. For the case that $G$ is connected, we add to the graph $G$ constructed in the first part of the proof a vertex $u$ and, for each path vertex $v_i$, an edge $(v_i,u)$. To $P$ we add a vertex $u'$ that has an edge to some path vertex of each component. By virtue of their high degrees, $u$ must be mapped to $u'$ and the remainder of the reduction proceeds in the same way as in the first part of the proof.
\qed
\end{proof}

We now show how to adapt this hardness proof to the case of {\sc Induced Subgraph}:

\begin{theorem}
Suppose the Exponential Time Hypothesis holds. Then there is no algorithm solving {\sc Induced Subgraph} in $2^{o(n/\log n)}$ time, even if $G$ is a caterpillar tree of maximum degree $3$ or $G$ is connected, planar, has pathwidth $2$ and has only one vertex of degree greater than $3$ and $P$ is a tree.
\end{theorem}

\begin{proof}
Matou{\v{s}}ek and Thomas \cite{isotreewidth} observe that by subdividing each edge once, a subgraph problem becomes an induced subgraph problem. Due to the nature of our construction, we do not need to subdivide the hair edges. We can adapt the proof of Theorem \ref{thm:sgihard} by subdividing every path edge (but not the hair edges) and for the connected case, also subdividing the edges that connect to the central vertices $u$ and $u'$.\qed
\end{proof}

We now show how to adapt this proof to {\sc (Induced) Minor, Shallow Minor} and {\sc Topological Minor}:

\begin{theorem}
Suppose the Exponential Time Hypothesis holds. Then there is no algorithm solving {\sc (Induced) Minor, Shallor Minor} or {\sc Topological Minor} in $2^{o(n/\log n)}$ time, even if $G$ is a caterpillar tree of maximum degree $3$ or $G$ is connected, planar, has pathwidth $2$ and has only one vertex of degree greater than $3$ and $P$ is a tree.
\end{theorem}

\begin{proof}
We can use the same reduction as for (induced) subgraph. Clearly, if $P$ is an (induced) subgraph of $G$ then it is also an (induced/shallow/topological) minor. Conversely, if $P$ is not a subgraph of $G$ then allowing contractions in $G$ or subdivisions in $P$ do not help: contracting a hair edge simply removes that edge and vertex from the graph, while contracting a path edge immediately makes the path too short to fit all the components.\qed
\end{proof}
\FloatBarrier
\section{Tree and Path Decompositions with Few Bags}

In this section, we study the minimum size tree and path decomposition problems:

\begin{verse}
{\sc Minimum Size Tree Decomposition of width $k$ ($k$-MSTD)}\\
{\bf Given:} A graph $G$, integers $k,n$. \\
{\bf Question:} Does $G$ have a tree decomposition of width at most $k$, that has at most $n$ bags?
\end{verse}

The Minimum Size Path Decomposition ($k$-MSPD) problem is defined analogously. The following theorem is an improvement over Theorem 3 of \cite{mspd-springer}, where the same was shown for $k\geq 39$; our proof is also simpler.

\begin{theorem}\label{thm:mspd}
Let $k \geq 16$. Suppose the Exponential Time Hypothesis holds, then there is no algorithm for $k$-MSPD or $k$-MSTD using $2^{o(n/\log n)}$ time.
\end{theorem}

\begin{proof}
By reduction from {\sc Orthogonal Vector Crafting}. We begin by showing the case for MSPD, but note the same reduction is used for MSTD.

For the string $s$, we create a connected component in the graph $G$ as follows: for $1\leq i \leq |s|+1$ we create a clique $C_i$ of size $6$, and (for $1\leq i \leq|s|$) make all vertices of $C_i$ adjacent to all vertices of $C_{i+1}$. For $1 \leq i \leq |s|$, if $s(i)=1$, we create a vertex $s_i$ and make it adjacent to the vertices of $C_i$ and $C_{i+1}$.

For each string $t_i$, we create a component in the same way as for $s$, but rather than using cliques of size $6$, we use cliques of size $2$: for each $1\leq i\leq n$ and $1\leq j\leq |t_i|+1$ create a clique $T_{i,j}$ of size $2$ and (for $1\leq j\leq |t_i|$) make all vertices of $T_{i,j}$ adjacent to all vertices of $T_{i,j+1}$. For $1\leq j\leq |t_i|$, if $t_i(j)=1$, create a vertex $t_{i,j}$ and make it adjacent to the vertices of $T_{i,j}$ and $T_{i,j+1}$.

An example of the construction (for $s=10110$ and $t_1=01001$) is shown in Figure \ref{fig:mstd-red}. We now ask whether a path decomposition of width $16$ exists with at most $|s|$ bags.

\begin{figure}[h]
    \centering
    \includegraphics[width=0.7\textwidth]{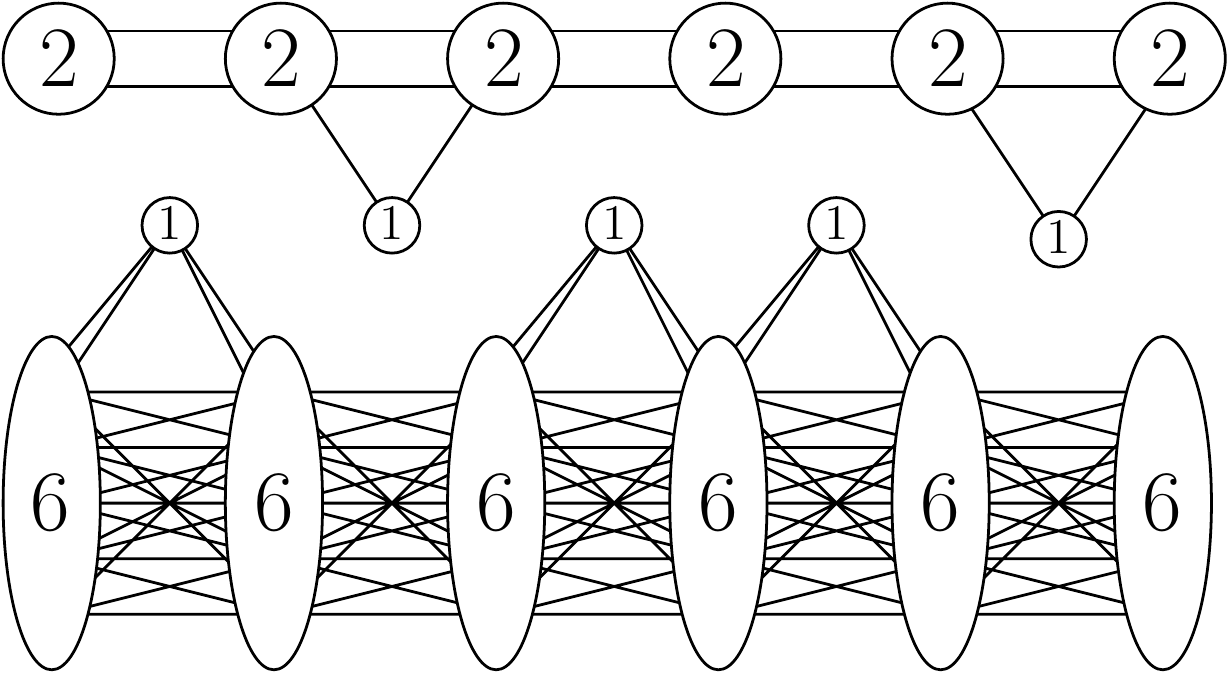}
    \caption{Simplified example of the graph created in the hardness reduction for Theorem \ref{thm:mspd}. The circles and ellipses represent cliques of various sizes. The component depicted in the top of the picture corresponds to $t_1=01001$, while the component at the bottom corresponds to $s=10110$.}
    \label{fig:mstd-red}
\end{figure}
\ifappendix For the correctness proof of the reduction, we refer to Appendix \ref{app:mspd-lab}.\fi

\begin{maybeappendix}{app:mspd}
\begin{lemma}
If there exists a solution $\Pi$ to the {\sc Orthogonal Vector Crafting} instance, then $G$ has a path decomposition of width $16$ with at most $|s|$ bags.
\end{lemma}

\begin{proof}
Given a solution $\Pi$, we show how to construct such a decomposition with bags $X_i,1\leq i\leq |s|$. In bag $X_i$ we take the vertices $C_i,C_{i+1}$ and (if it exists) the vertex $s_i$. We also take the cliques $T_{idx_\Pi(i), pos_\Pi(i)}$, $T_{idx_\Pi(i+1), pos_\Pi(i+1)}$ and the vertex $t_{idx_\Pi(i),pos_\Pi(i)}$ (if it exists).

Each bag contains two cliques of size $6$ and two cliques of size $2$, adding up to $16$ vertices in each bag. Each bag may additionally contain a vertex $s_i$ or a vertex $t_{idx_\Pi(i),pos_\Pi(i)}$, but, by nature of a solution to {\sc Orthogonal Vector Crafting}, not both, showing that each bag contains at most $17$ vertices and as such the decomposition indeed has width $16$. \qed
\end{proof}

\begin{lemma}
If $G$ has a tree decomposition of width $16$ with at most $|s|$ bags, then the instance of {\sc Orthogonal Vector Crafting} has a solution.
\end{lemma}
\begin{proof}
Suppose we have a tree decomposition of width $16$ with at most $|s|$ bags. Since for any $1\leq i \leq|s|$, $C_i \cup C_{i+1}$ induces a clique, there exists a bag that contains both $C_i$ and $C_{i+1}$. Moreover, since each clique $C_i$ contains $6$ vertices, each bag of the decomposition can contain at most two such cliques. Moreover, since the decomposition has at most (and thus we can assume exactly) $|s|$ bags, there exists exactly one bag containing both $C_i$ and $C_{i+1}$.

Since (for $1<i<|s|$) the bag containing $C_i$ and $C_{i+1}$ must be adjacent to the bag containing $C_{i-1}$ and $C_{i}$ and to the bag containing $C_{i+1}$ and $C_{i+2}$ we see that all but two bags have degree at least two and the remaining two bags have degree at least 1. Since a tree/path decomposition can not have cycles, we see that the tree/path decomposition must take the form of a path, where the bag containing $C_i$ and $C_{i+1}$ is adjacent to the bag containing $C_{i-1}$ and $C_{i}$ and to the bag containing $C_{i+1}$ and $C_{i+2}$.

Assume the bags of the decomposition are $X_1,\ldots X_s$, and $C_{i},C_{i+1}\subseteq X_i$.

Since in each of the bags we now have capacity for $5$ more vertices, we see that each bag contains exactly two cliques $T_{i,j},1\leq i\leq n, 1\leq j\leq |t_i|+1$ and that there exists a bag that contains both $T_{i,j}$ and $T_{i,j+1}$ for all $1\leq i\leq n, 1\leq j \leq |t_i|$. Moreover, for each $i$, the bags containing $\{T_{i,j} \cup T_{i,j+1} : 1\leq j \leq |t_i|\}$ must be consecutive and appear in that order (or in the reverse order, but by the palindromicity of $t_i$ this case is symmetric).

Note that at this point, each bag contains exactly $16$ vertices and has room for exactly 1 more vertex (which can be either an $s_i$ or a $t_{ij}$).

Thus, if we were to list the intersection of each bag $X_i$ with $\{T_{i,j}, 1\leq i\leq n, 1\leq j\leq |t_i|\}$, we would obtain the following sequence:
 
 $$ \{T_{\Pi(1), 1}, T_{\Pi(1), 2}\}, \ldots,  \{T_{\Pi(1), |t_{\Pi(1)}|}, T_{\Pi(1), |t_{\Pi(1)}|+1}\}$$
 $$\{T_{\Pi(2), 1}, T_{\Pi(2), 2}\}, \ldots,  \{T_{\Pi(2), |t_{\Pi(2)}|}, T_{\Pi(2), |t_{\Pi(2)}|+1}\}$$
 $$ \ldots$$ 
 $$ \{T_{\Pi(n), 1}, T_{\Pi(n), 2}\}, \ldots,  \{T_{\Pi(n), |t_{\Pi(n)}|}, T_{\Pi(n), |t_{\Pi(n)}|+1}\}$$

Which gives us the permutation $\Pi$ and string $t^\Pi$ we are after. Note that $s(i)$ and $t^\Pi(i)$ can not both be $1$, because otherwise the vertex $s_i$ (being adjacent to the vertices of $C_{i}$ and $C_{i+1}$) must go in bag $X_i$, but so must the vertex $t_{idx_\Pi(i),pos_\Pi(i)}$ which would exceed the capacity of the bag. \qed
\end{proof}

The size of the graph created in the reduction is $O(|s|)$, so we obtain a $2^{\Omega(n/\log n)}$ lower bound for $16$-MSPD under the Exponential Time Hypothesis. We can extend the reduction to $k>16$ by adding universal vertices to the graph.

For the tree decomposition case, note that a path decomposition is also a tree decomposition. For the reverse implication, we claim that a tree decomposition of $G$ of width $16$ with at most $|s|$ bags must necessarily be a path decomposition. This is because for each $1\leq i\leq |s|$, there exists a unique bag containing $C_i$ and $C_{i+1}$ which is adjacent to the bag containing $C_i-1$ and $C_{i}$ (if $i > 1$) and to the bag containing $C_{i+1}$ and $C_{i+2}$ (if $i<|s|$). All but two bags are thus adjacent to at least two other ones, and a simple counting argument shows that there therefore is no bag with degree greater than two (or we would not have enough edges).\qed
\end{maybeappendix}
\end{proof}
\FloatBarrier
\section{Intervalizing Coloured Graphs}

In this section, we consider the problem of intervalizing coloured graphs:

\begin{verse}
{\sc Intervalizing Coloured Graphs}\\
{\bf Given:} A graph $G=(V,E)$ together with a proper colouring $c:V\to \{1,2,\ldots,k\}$. \\
{\bf Question:} Is there an interval graph $G'$ on the vertex set $V$, for which $c$ is a proper colouring, and which is a supergraph of $G$?
\end{verse}

{\sc Intervalizing Coloured Graphs} is known to be $\NP$-complete, even for 4-coloured caterpillars (with hairs of unbounded length) \cite{4colcaterpillar}. In contrast with this result we require five colours instead of four, and the result only holds for trees instead of caterpillars. However, we obtain a $2^{\Omega(n/\log n)}$ lower bound under the Exponential Time Hypothesis, whereas the reduction in \cite{4colcaterpillar} is from {\sc Multiprocessor Scheduling} and to our knowledge, the best lower bound obtained from it is $2^{\Omega(\sqrt[5]{n})}$ (the reduction is weakly polynomial in the length of the jobs, which following from the reduction from {\sc 3-Partition} in \cite{3partitioneth} is $\Theta(n^4)$). In contrast to these hardness results, for the case with $3$ colours there is an $O(n^2)$ time algorithm \cite{UUCS199515,babette-journal}.

\begin{theorem}\label{thm:intervalizing}
{\sc Intervalizing Coloured Graphs} does not admit a $2^{o(n/\log n)}$-time algorithm, even for $5$-coloured trees, unless the Exponential Time Hypothesis fails.
\end{theorem}

\begin{proof}
Let $s,t_1,\ldots,t_n$ be an instance of {\sc Orthogonal Vector Crafting}. We construct $G=(V,E)$ in the following way:

\paragraph*{S-String Path.} We create a path of length $2|s|-1$ with vertices $p_0,\ldots p_{2|s|-2}$, and set $c(p_i)=1$ if $i$ is even and $c(p_i)=2$ if $i$ is odd. Furthermore, for even $0\leq i\leq 2|s| - 2$, we create a neighbour $n_i$ with $c(n_i)=3$.

\paragraph*{Barriers.} To each endpoint of the path, we attach the \emph{barrier gadget}, depicted in Figure \ref{fig:barrier}. The gray vertices are not part of the barrier gadget itself, and represent $p_0$ and $n_0$ (resp. $p_{2|s|-2}$ and $n_{2|s|-2}$). Note that the barrier gadget operates on similar principles as the barrier gadget due to Alvarez et al. \cite{4colcaterpillar}. We shall refer to the barrier attached to $p_0$ as the left barrier, and to the barrier attached to $p_{2|s|-2}$ as the right barrier.

The barrier consists of a central vertex with colour $1$, to which we connect eight neighbours (\emph{clique vertices}), two of each of the four remaining colours. Each of the clique vertices is given a neighbour with colour $1$. To one of the clique vertices with colour $2$ we connect a vertex with colour $3$, to which a vertex with colour $2$ is connected (\emph{blocking vertices}). This clique vertex shall be the barrier's \emph{endpoint}. Note that the neighbour with colour $1$ of this vertex is not considered part of the barrier gadget, as it is instead a path vertex. We let $C_l$ ($e_l)$ denote the center (endpoint) of the left barrier, and $C_r$ ($e_r$) the center (endpoint) of the right barrier.

\begin{figure}[h]
     \centering
     \hfill
     \subfloat[][Barrier Gadget] {
         \includegraphics[scale=0.9]{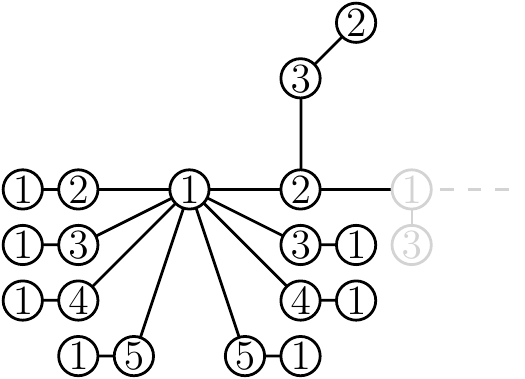}
     }
     \hfill
     \subfloat[][Interval Representation] {
         \includegraphics[scale=0.9]{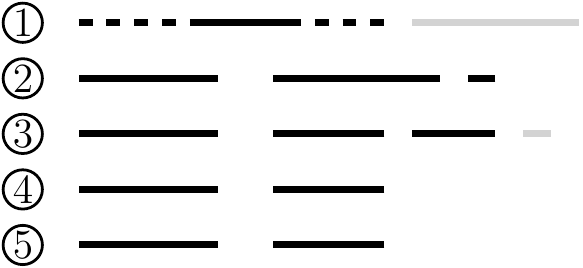}
     }
     \hfill\null
     \caption{(a) Barrier Gadget. The gray vertices are not part of the barrier gadget itself, and show how it connects to the rest of the graph. (b) How the barrier gadget may (must) be intervalized.}
     \label{fig:barrier}
\end{figure}

\paragraph*{T-String Paths.} Now, for each string $t_i$, we create a path of length $2|t_i|+1$ with vertices $q_{i,0},\ldots,q_{i,2|t_i|}$ and set $c(q_{i,j})=3$ if $j$ is odd and set $c(q_{i,j})=2$ if $j$ is even. We make $q_{i,1}$ adjacent to $U$. Furthermore, for odd $1\leq j\leq 2|t_i|-1$, we create a neighbour $m_j$ with $c(m_j)=1$. We also create two \emph{endpoint vertices} of colour $3$, one of which is adjacent to $q_{i,0}$ and the other to $q_{i,2|t_i|}$,

\paragraph*{Connector Vertex.} Next, we create a \emph{connector vertex} of colour $5$, which is made adjacent to $p_1$ and to $q_{i,1}$ for all $1\leq i\leq n$. This vertex serves to make the entire graph connected.

\paragraph*{Marking Vertices.} Finally, for each $1\leq i\leq |s|$ (resp. for each $1\leq i\leq n$ and $1\leq j\leq |t_i|$), if $s(i)=1$ (resp. $t_i(j)=1$), we give $p_{2i-1}$ (resp. $q_{i,2j-1}$) two neighbours (called the \emph{marking vertices}) with colour $4$. For each of the marking vertices, we create a neighbour with colour $3$.

This construction is depicted in Figure \ref{fig:intervalizing-red}. In this example $s=10100$, $t_1=01$ and $t_2=001$. Note that this instance of {\sc Orthogonal Vector Crafting} is illustrative, and does not satisfy the restrictions required in the proof.

\begin{figure}[h]
    \centering
    \includegraphics[width=0.95\textwidth]{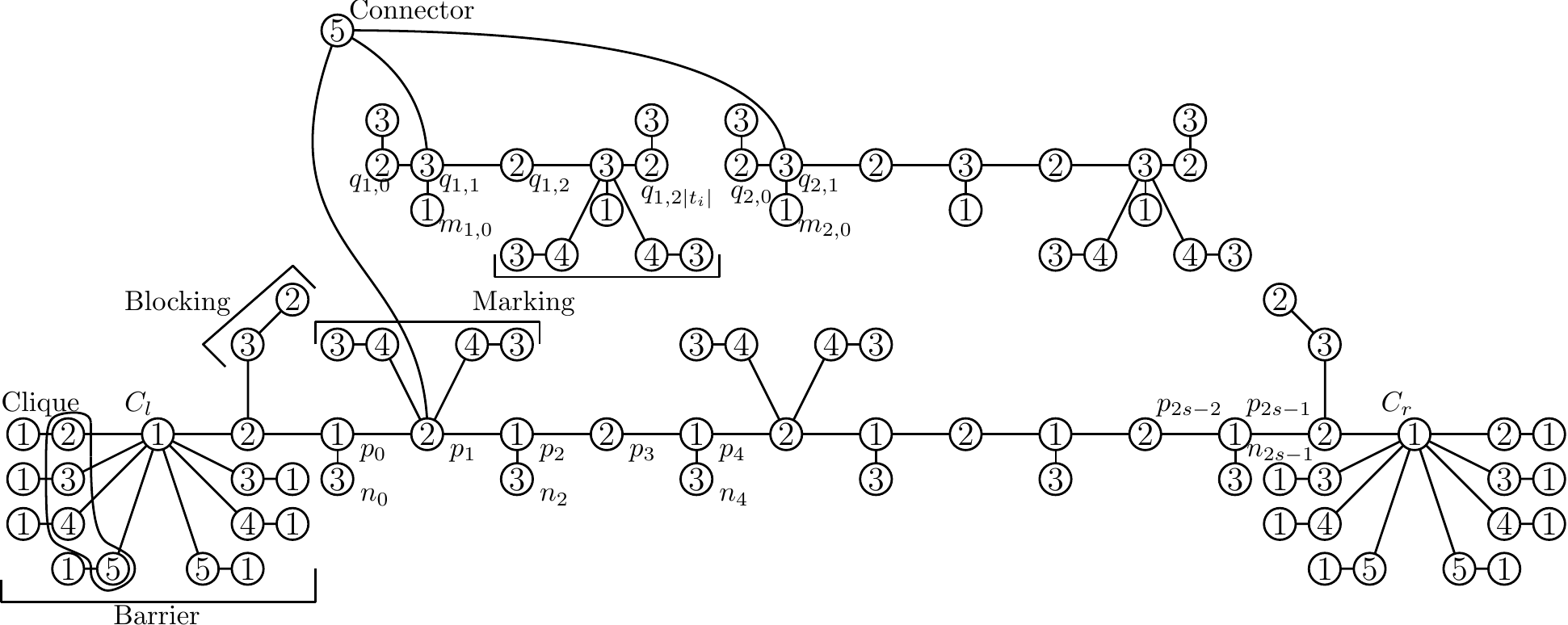}
    \caption{Example of the graph created in the hardness reduction for Theorem \ref{thm:intervalizing}.}
    \label{fig:intervalizing-red}
\end{figure}

Informally, the construction works as follows: the barriers at the end of the path of $p$-vertices can not be passed by the remaining vertices, meaning we have to "weave" the shorter $q$-paths into the long $p$-path. The colours enforce that the paths are in "lockstep", that is, we have to traverse them at the same speed. We have to map every $q$-vertex with colour $3$ to a $p$-vertex with colour $2$, but the marking vertices prevent us from doing so if both bitstrings have a $1$ at that particular position.
\ifappendix See Appendix~\ref{app:fewcol} for the formal proof of correctness.\fi

\begin{maybeappendix}{app:fewcol}
\begin{lemma}\label{lem:invervalizing-if}
$G$ can be intervalized if the {\sc Orthogonal Vector Crafting} instance has a solution.
\end{lemma}

\begin{proof}
As an example, Figure \ref{fig:intervalizing-solution} shows how the graph from Figure \ref{fig:intervalizing-red} may be intervalized. Let $\Pi$ be a solution to the instance of {\sc Orthogonal Vector Crafting}. We can intervalize the barriers as depicted in Figure \ref{fig:barrier}b, noting that the right barrier should be intervalized in a mirrored fashion. The connector vertex (which is the only remaining vertex of colour $5$) can be assigned an interval that covers the entire interval between the two barrier gadgets. If no marker vertices are present, then we can weave the $q$-paths into the $p$-path as depicted in Figure \ref{fig:intervalizing-solution}, whereby each interval corresponding to a $q$-vertex of colour $3$ completely contains the interval of a $p$-vertex of colour $2$ (note that the endpoint vertices of the $q$-paths are treated differently from the $q$-path vertices with colour $3$). If we intervalize the $q$-paths in the same order the corresponding strings appear in $\Pi$, then we can also intervalize the marking vertices: Figure \ref{fig:intervalizing-solution} shows that the marker vertices can also be intervalized, so long as a $p$-vertex and its corresponding $q$-vertex are not both adjacent to marker vertices, but this is guaranteed by the orthogonality of $s$ and $t^\Pi$.
\qed
\end{proof}

\begin{figure}[b]
    \centering
    \includegraphics[width=0.95\textwidth]{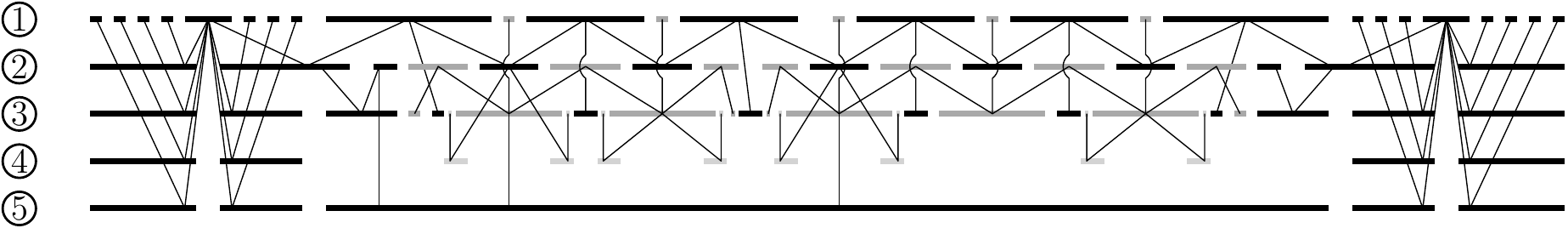}
    \caption{How the graph from Figure \ref{fig:intervalizing-red} may be intervalized. Thick lines represent intervals and are arranged vertically based on the colour of the associated vertex. The thinner lines indicate the edges of the original graph. Black intervals correspond to the barriers, $p,n$-vertices and the connector vertex, light gray intervals to maker vertices and dark gray intervals to end point vertices and to $q,m$-vertices.}
    \label{fig:intervalizing-solution}
\end{figure}

\begin{lemma}\label{lem:invervalizing-onlyif}
$G$ can be intervalized only if the {\sc Orthogonal Vector Crafting} instance has a solution.
\end{lemma}

\begin{proof}
Suppose we are given a properly coloured interval supergraph of $G$ that assigns to each vertex $v\in V$ its left endpoint $l(v)$ and right endpoint $r(v)$. For vertices $u,v\in V$, we write $v\subset u$ if $l(u) < l(v) < r(v) < r(u)$, and we write $v<u$ if $r(v)<l(u)$. We write $v\lessdot u$ if $l(v) < l(u) < r(v) < r(u)$ - that is, the interval of $v$ starts to the left of the interval of $u$ and the two intervals overlap.

We may without loss of generality assume that $C_l < C_r$ and that no two endpoints of intervals coincide.

\begin{claim}
For any non-barrier vertex $v$, we have that $r(C_l) < l(v) < r(v) < l(C_r)$.
\end{claim}

\begin{proof}
Examining the situation for the left barrier, we see that a clique vertex can not be contained completely inside the interval of the center vertex $C_l$ since it is adjacent to another vertex with colour $1$ (whose interval may not intersect that of the center). Since there are two clique vertices of each colour, for each colour, the interval of one clique vertex of that colour must extend to the left of the central vertex' interval and the other must extend to the right. Therefore, these intervals that contain $r(c_l)$ induce a clique of size $5$. Since we are looking for a $5$-coloured supergraph, no other intervals can contain $r(c_l)$.

Note that the clique vertices are interchangeable, except the ones coloured $2$: the clique vertex that is adjacent to $p_0$ must go to the right of the center vertex, since otherwise the path from it to $C_r$ could not pass the clique at $r(c_l)$.

Suppose that for some non-barrier vertex $v$, it holds that $l(v) \leq r(c_r)$. This is not possible, since the path from $v$ to $c_r$ can not pass the clique induced at $r(c_l)$. Therefore, $r(c_l) < l(v)$.

The case for the right barrier is symmetric. \qed
\end{proof}

\begin{claim}
For all $0\leq i < 2|s|-1$, we have that $p_i \lessdot p_{i+1}$. Furthermore, $C_l\lessdot e_l \lessdot p_0$ and $p_{2|s|-1}\lessdot e_r\lessdot C_r$.
\end{claim}

\begin{proof}
The fact that $C_l\lessdot e_l \lessdot p_0$ (resp. $p_{2|s|-1}\lessdot e_r\lessdot C_r$) follows from the analysis of the barrier gadget in the previous claim. We now proceed by induction, for the purposes of which we shall write $e_l=p_{-1}$ and $e_r=p_{2|s|}$: suppose the claim holds for $i-1$, i.e. $p_{i-1}\lessdot p_{i}$. It can not hold that $p_{i+1}\subset p_i$, since it is adjacent to $p_{i+2}$, nor can it hold that $p_{i+1}\lessdot p_i$ (since then it would intersect the interval of $p_{i-1}$. \qed
\end{proof}

For any $1\leq i\leq n$, a similar fact holds for the path $q_{i,0},\ldots,q_{i,2|t_i|}$:

\begin{claim}
We may without loss of generality assume that $q_{i,0}\lessdot \ldots \lessdot q_{i,2|t_i|}$.
\end{claim}

\begin{proof}
By a similar induction argument as above, it must hold that either $q_{i,0}\lessdot \ldots \lessdot q_{i,2|t_i|}$ or $q_{i,0}\gtrdot \ldots \gtrdot q_{i,2|t_i|}$. However, since $t_i$ is palindromic, these two cases are equivalent. \qed
\end{proof}

\begin{claim}
Let $1\leq i \leq n, 1\leq j\leq |t_i|$. Then there exists $1\leq k\leq |s|$ such that $p_{2k-1} \subset q_{i,2j-1}$.
\end{claim}

\begin{proof}
The interval of $q_{i,2j-1}$ can not be completely contained in the interval of a vertex with colour $1$, since $q_{i,2j-1}$ is adjacent to $m_{i,2j-1}$ which has colour $1$ as well. Therefore (since $r(C_l) < l(q_{i,2j-1}) < r(q_{i,2j-1}) < l(C_r)$), the interval of $q_{i,2j-1}$ must intersect the interval of either a barrier endpoint or a path vertex $p_{2k-1}$ for some $1\leq k\leq |s|$. It is not possible that $q_{i,2j-1}$ intersects the interval of a barrier endpoint, since (due to the blocking and clique vertices with colour $3$, see Figure \ref{fig:barrier}) it would have to be completely contained inside this interval, which is impossible since $q_{i,2j-1}$ is adjacent to vertices with colour $2$. Therefore there exists a $1\leq k\leq |s|$ such that the interval of $q_{i,2j-1}$ intersects that of $p_{2k-1}$.

Since $q_{i,2j-2} \lessdot q_{i,2j-1} \lessdot q_{i,2j}$ and $c(q_{i,2j-2})=c(q_{i,2j})=c(p_{2k-1})=2$ we must have that $q_{i,2j-2}<q_{2k-1}<q_{i,2j}$. It now follows that $p_{2k-1} \subset q_{i,2j-1}$. \qed
\end{proof}

This allows us to define the \emph{position} $1\leq P(i,j)\leq |s|$ for each $1\leq i \leq n, 1\leq j\leq |t_i|$, which is equal to the $k$ from the previous claim. Note that $P$ is a bijection, since each interval $p_{2k-1}$ is the subset of (the interval of) exactly one $q$-vertex (each $p$-vertex interval can not be the subset of more than one $q$-vertex interval, and the number of $p$-vertices is such that no $q$-vertex can completely contain more than one $q$-vertex interval).

\begin{claim}
Let $1\leq i \leq n, 1\leq j\leq |t_i|-1$. Then $P(i,j+1)=P(i,j)+1$. 
\end{claim}

\begin{proof}
This follows from the fact that a $q$-vertex with colour $3$ can not completely contain a $p$-vertex with colour $1$ (since it has an $n$-vertex as neighbour that has colour $1$ as well) and the fact that a $q$-vertex with colour $2$ can not overlap a $p$-vertex with colour $2$.

Formally, let $P(i,j) = k$, then $p_{2k-1}\subset q_{i,2j-1}$. Since $q_{i,2j-1}\lessdot q_{i,2j}\lessdot q_{i,2j+1}$, if the claim does not hold, we must have $p_{2k+1+2m} \subset q_{i,2j+1}$ for some $m > 0$. We must have $q_{i,2j}\subset p_{2k+2r}$ for some $r \geq 0$. If $r=0$, then $q_{i,2j} < q_{i,2j+1}$ which is a contradiction. On the other hand, if $r > 0$ then $q_{i,2j - 1} < q_{i,2j}$ which is also a contradiction. \qed
\end{proof}

\begin{claim}
Let $\Pi$ be the permutation of $\{1,\ldots,n\}$ such that $P(idx_\Pi(i),pos_\Pi(i))=i$. Then $\Pi$ exists and is a solution to the {\sc Orthogonal Vector Crafting} instance.
\end{claim}
\begin{proof}
The existence of $\Pi$ follows from the previous claim. Suppose that $\Pi$ is not a solution. Then there exists an $i$, such that $S(i)=T_{idx_\Pi(i)}(pos_\Pi(i))=1$. However, this means that $p_{2i-1}\subset q_{idx_\Pi(i),2pos_\Pi(i)-1}$. Since both $p_{2i-1}$ and $q_{idx_\Pi(i),2pos_\Pi(i)-1}$ have marking vertices, this is impossible as the marking vertices with colour $4$ would have overlapping intervals. \qed
\end{proof}

This completes the proof of Lemma \ref{lem:invervalizing-onlyif}. \qed
\end{proof}
\end{maybeappendix}

The number of vertices of $G$ is linear in $|s|$, and we thus obtain a $2^{o(n/\log n)}$ lower bound under the Exponential Time Hypothesis. \qed
\end{proof}

Note that the graph created in this reduction only has one vertex of super-constant degree. This is tight, since the problem is polynomial-time solvable for bounded degree graphs (for any fixed number of colours) \cite{boundeddegree}.

To complement this result for a bounded number of colours, we  also show 
a $2^{\Omega(n)}$-time lower bound for graphs with an unbounded number of colors, assuming the ETH.  Note that this result implies that the algorithm from \cite{intervalizing-exact} is optimal. \ifappendix\else A complication in the proof is that to obtain the stated bound, one can only use (on average) a constant number of vertices of each colour (when using $O(n)$ colours). A variation on the previous proof whereby instead of using bitstrings, colours are used to identify clauses and variables is thus unlikely to work since one would need to repeat each colour many times (in each place where a particular bitstring does not fit).\fi

\ifappendix \begin{theorem}[Appendix \ref{app:manycol}.]\else \begin{theorem}\fi\label{thm:manycols}
Assuming the Exponential Time Hypothesis, there is no algorithm solving {\sc Intervalizing Coloured Graphs} in time $2^{o(n)}$, even when restricted to trees.
\end{theorem}

\begin{maybeappendix}{app:manycol}
\begin{proof}\ifappendix \emph{(Theorem \ref{thm:manycols}.)} \fi
By reduction from {\sc Exact Cover by 3-Sets (X3C)}. {\sc X3C} is the following problem: given a set $X$ with $|X|=n$ and a collection $M$ of subsets $X_1,\ldots, X_m$ of size $3$, decide whether there exists a subset $M'$ of size $n/3$ such that $\bigcup M' = X$. Assuming the Exponential Time Hypothesis, there does not exist an algorithm solving {\sc X3C} in time $2^{o(m)}$ (See e.g. \cite{mspd-springer,vectorscheduling,garey-johnson}. Note that {\sc 3-Dimensional Matching} is a special case of {\sc X3C}.).

The intuition behind the reduction is that, to leverage the large number of colours, one needs to force the creation of a very large clique in the intervalized graph. For each element of $X_i\in M$, we create a component of which the vertices have colours that correspond to the elements of $X_i$. The graph created in the reduction will be so that all but $n/3$ of the components corresponding to elements of $M$ can be placed in distinct intervals, but the remaining $n/3$ components will have to overlap. This, in turn, is only possible if no two components in this latter collection contain duplicated colours, that is, each element is represented at most (and thus exactly) once in the selected $n/3$ components.

We assume the elements of $X$ are labelled $1,\ldots,n$. We may assume that $n$ is a multiple of $3$ (or we may return a trivial no-instance) and that $m > n / 3$ (or we can check immediately whether $M$ covers $X$).

The graph created in the reduction has $n+4$ colours: two colours $e_i$ and $f_i$ for each $1\leq i \leq n$ and four additional colours $a,b,c,d$.

We construct the graph $G$ as follows: we start with a path of $2(m-n/3) + 1$ vertices $p_0,\ldots,p_{2(m-n/3)}$, where $p_i$ has colour $a$ if $i$ is even and $p_1$ has colour $b$ if $i$ is odd. To $p_{2(m-n/3)}$ we make adjacent a vertex with colour $d$, which we shall refer to as the \emph{right barrier} vertex.

Next, for each $1\leq i\leq n$, we create a vertex $v_i$ with colour $e_i$, that is made adjacent to $p_0$. For each $i$, we create an additional vertex with colour $d$ that is made adjacent to $v_i$. These vertices (with colour $d$) are called the \emph{left barrier} vertices.

Next, for each 3-set $X_i\in M$, we create a \emph{set component}, consisting of a \emph{central vertex} with colour $a$, to which are made adjacent: two vertices with colour $c$, each of which is adjacent to a vertex with colour $b$ and, for each of the three elements of $X_i$, two vertices with the corresponding (to that element) colour $f_i$, each of which is made adjacent to a vertex with colour $e_i$.

Finally, we connect all the components together with a \emph{connector vertex} of colour $d$, which is made adjacent to each set gadget and to a vertex of the path (for instance to $p_0$). Figure \ref{fig:intervalizing-many} provides an example of the construction.

As with the case for five colours, a solution somehow has to ``pack'' the set components into the part of the graph between the barriers: $m-n/3$ of the set components can be packed into the path (each interval corresponding to a $b$-vertex can hold at most one component); the remaining $n/3$ components have to be packed between the left barrier vertices and $p_0$ and this is possible only if the corresponding sets are disjoint (otherwise intervals of the corresponding colours would have to overlap).

\begin{figure}[h]
    \centering
    \includegraphics[width=0.95\textwidth]{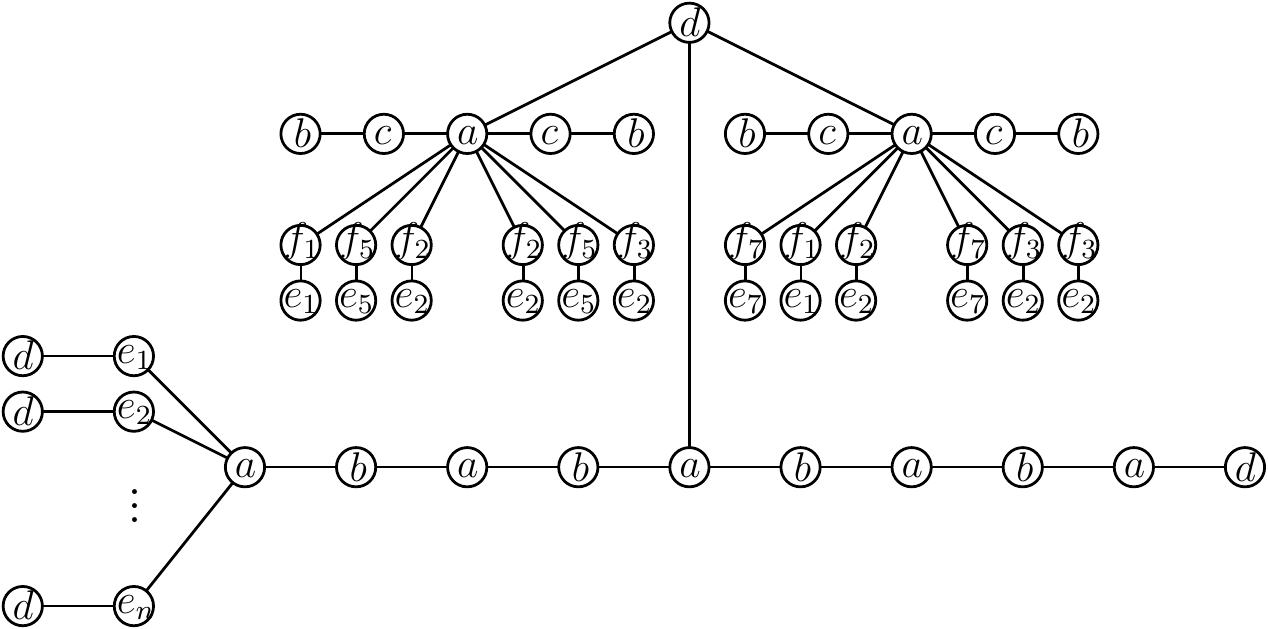}
    \caption{The construction used in the proof of Theorem \ref{thm:manycols}. In this example, $X=\{\{1,5,2\},\{7,1,2\},\ldots\}$. The topmost vertex with colour $d$ is the connector vertex, and its degree would increase as more elements are added to $X$.}
    \label{fig:intervalizing-many}
\end{figure}

\begin{lemma}
If $G$ can be intervalized, then the {\sc X3C} instance has a solution.
\end{lemma}

\begin{proof}
In any interval supergraph of $G$, the interval of the connector vertex should lie between the intervals of the left barrier vertices and the right barrier vertex. Let $X_i\in M$. The interval of the central vertex of the set component corresponding to $X_i$ must either be contained in the interval corresponding to some vertex $p_i$ for some odd $1\leq i < 2(m-n/3)$ or it should be contained in the intersection of the intervals corresponding to the vertices $\{v_j \mid 1\leq j\leq n\}$ (since its interval can not intersect any interval of a vertex with colour $a$, nor can its interval be contained in the interval of a left or right barrier vertex).

\begin{claim}
At most one interval corresponding to a central vertex can be contained in the interval of any vertex $p_i$ (for odd $i$).
\end{claim}
\begin{proof}
Each central vertex is adjacent to two vertices with colour $c$. Since these vertices in turn have neighbours with colour $b$, it follows that (if some central vertex is contained in the interval of $p_i$) the interval of one of the vertices with colour $c$ must contain the left endpoint of the interval of $p_i$, and the other must contain the right endpoint. Therefore the interval of $p_i$ can not contain another central vertex. \qed
\end{proof}

\begin{claim}
Let $X_i\not = X_j\in M$. If $X_i\cap X_j\not = \emptyset$, then the intervals of the central vertices of the set components corresponding to $X_i$ and $X_j$ can not both be contained in the intersection of the intervals corresponding to the vertices $\{v_k \mid 1\leq k\leq n\}$.
\end{claim}
\begin{proof}
Since $X_i\cap X_j\not = \emptyset$, both central vertices have two neighbours with color $f_m$ for some $1\leq m\leq n$. Each of these vertices has a neighbour with colour $e_m$, and thus the interval of a vertex with colour $f_m$ must contain either the left or the right endpoint of the interval that is the intersection of the intervals corresponding to the vertices $\{v_k \mid 1\leq k\leq n\}$. This is not possible, since either the left or the right endpoint of this intersection will be contained in more than one interval corresponding to a vertex with colour $f_m$. \qed
\end{proof}

We thus see that the elements of $M$ that correspond to set components whose intervals are contained in the intersection of the intervals corresponding to the vertices $\{v_k \mid 1\leq k\leq n\}$ form a solution: no element of $X$ is contained in more than one of them. As at most $|X|-n/3$ intervals of set components are contained in intervals corresponding to some vertex $v_i$, we have that at least $n/3$ set components have intervals that are contained in the aforementioned intersection. These must thus form a solution to the {\sc X3C} instance. \qed
\end{proof}

\begin{lemma}
If the {\sc X3C} instance has a solution, then $G$ can be intervalized.
\end{lemma}

\begin{proof}
The $v$-vertices are assigned to identical intervals. Their $d$-coloured neighbours can be assigned arbitrarily small intervals that are placed in the left half of the $v$-vertex interval. The $p$-vertices can then be placed from left to right, so that $p_0$ overlaps a small portion of the right end of the $v$-vertex intervals, and each $p$ vertex interval overlaps the interval of the preceding $p$-vertex slightly. Finally the right barrier vertex (with colour $d$) should be placed in the right half of the interval corresponding to $p_2{2(m-n/3)}$.

Next, the connector vertex (with colour $d$) can be assigned a long interval between the left and right barrier vertices, that overlaps the $v$-vertex intervals and all of the $p$-vertex intervals. This placement of the connector vertex allows us to place the intervals of set components anywhere between the left and right barriers.

Let $M'\subseteq M$ be a solution to the {\sc X3C} instance. Since $|M'|=n/3$, and there are $m-n/3$ $p$-vertices with colour $b$, we can assign each element of $M$ that is not in the solution a unique $p$-vertex with colour $a$. We assign the central vertex (which has colour $a$) of the set component of each such element an interval inside the interval of its $p$-vertex, such that it does not intersect neighbouring $p$-vertices (which have color $a$). The vertices with $f$- or $e$-colours of the set component can be assigned similar intervals (not intersecting neighbouring $p$-vertices). One of the vertices with colour $c$ is assigned an interval that extends past the left endpoint of the $p$-vertex interval (and thus intersects the preceding $p$-vertex interval), which allows us to assign its neighbour with colour $b$ an interval that is contained in the preceding $p$-vertex interval (and does not intersect any other $p$-vertex interval). The other vertices with colours $b$ and $c$ can be placed similarly on the right.

Finally, for the set components corresponding to elements of $M'$, we can assign the vertices with colours $a,b,c$ arbitrarily small intervals in the right half of the $v$-vertex intervals; the $f$-coloured vertices can be placed so that their intervals stick out beyond the right and left endpoints of the $v$-vertex intervals, so that the $e$-coloured vertices can be placed not overlapping the $v$-vertex intervals. The fact that $M'$ is a solution guarantees this can be done without any $e,f$-colours overlapping each other, since each such colour occurs exactly twice (one such pair of vertices can be placed on the left, the other on the right). \qed
\end{proof}

This completes the reduction. Since the number of vertices of $G$ is linear in $|M|$ and $|X|$, we see that {\sc Intervalizing Coloured Graphs} does not admit a $2^{o(n)}$-time algorithm, unless the Exponential Time Hypothesis fails. \qed
\end{proof}
\end{maybeappendix}

\section{Conclusions}

In this paper, we have shown for several problems that, under the Exponential Time Hypothesis, $2^{\Theta(n/\log n)}$ is the best achievable running time - even when the instances are very restricted (for example in terms of pathwidth or planarity). For each of these problems, algorithms that match this lower bound are known and thus $2^{\Theta(n/\log n)}$ is (likely) the asymptotically optimal running time.

For problems where planarity or bounded treewidth of the instances (or, through bidimensionality, of the solutions) can be exploited, the optimal running time is often $2^{\Theta(\sqrt{n})}$ (or features the square root in some other way). On the other hand, each of problems studied in this paper exhibits some kind of ``packing'' or ``embedding'' behaviour.  For such problems, $2^{\Theta(n/\log n)}$ is often the optimal running time. We have introduced two artificial problems, {\sc String Crafting} and {\sc Orthogonal Vector Crafting}, that form a useful framework for proving such lower bounds.

It would be interesting to study which other problems exhibit such behaviour, or to find yet other types of running times that are ``natural'' under the Exponential Time Hypothesis. The loss of the $\log{n}$-factor in the exponent is due to the fact that $\log n$ bits or vertices are needed to ``encode'' $n$ distinct elements; it would be interesting to see if there are any problems or graph classes where a more compact encoding is possible (for instance only $\log^{1-\epsilon} n$ vertices required, leading to a tighter lower bound) or where an encoding is less compact (for instance $\log^2 n$ vertices required, leading to a weaker lower bound) and whether this can be exploited algorithmically.

\subsection*{Acknowledgement.} We thank Jesper Nederlof for helpful comments and discussions.

\bibliographystyle{splncs}

\bibliography{references}
\ifappendix
\clearpage

\appendix

\section*{Appendix}

\section{An Algorithm for String Crafting}\label{app:sc-algo}
\putmaybeappendix{app:sc-algo}

\section{Additional proofs for Subgraph embedding problems}\label{app:subgraphs-lab}
\putmaybeappendix{app:subgraphs}

\section{Minimum Size Tree and Path Decompositions}\label{app:mspd-lab}
\paragraph*{Continuation of the proof of Theorem \ref{thm:mspd}.}
\putmaybeappendix{app:mspd}

\section{Intervalizing $5$-Coloured Trees}\label{app:fewcol}
\paragraph*{Continuation of the proof of Theorem \ref{thm:intervalizing}.}
\putmaybeappendix{app:fewcol}

\section{Intervalizing Coloured Graphs with Many Colours}\label{app:manycol}
\putmaybeappendix{app:manycol}
\fi

\end{document}